\documentclass[a4paper,USenglish,autoref,cleveref]{lipics-v2021}
\usepackage[utf8]{inputenc}
\usepackage{amsthm,amsmath,amssymb}
\usepackage{graphicx}
\usepackage{xspace}
\usepackage{complexity}
\usepackage{xcolor}
\usepackage{verbatim}
\usepackage{booktabs}
\usepackage{tikz}

\title{Eliminating Crossings in Ordered Graphs}

\author{Akanksha~Agrawal}{Indian Institute of Technology Madras, Chennai, India}{}{0000-0002-0656-7572}{}

\author{Sergio~Cabello}{Faculty of Mathematics and Physics, University of Ljubljana, Ljubljana, Slovenia \and Institute~of~Mathematics, Physics and Mechanics, Ljubljana, Slovenia}{}{0000-0002-3183-4126}{}

\author{Michael~Kaufmann}{Department of Computer Science, Tübingen University, Tübingen, Germany}{}{}{}
\author{Saket~Saurabh}{Institute of Mathematical Sciences, Chennai, India}{}{}{}
\author{Roohani~Sharma}
{University of Bergen, Bergen, Norway}
{}{}{}
\author{Yushi~Uno}{Graduate School of Informatics, Osaka Metropolitan University, Sakai, Japan}{}{}{}

\author{Alexander~Wolff}{Universit\"at W\"urzburg, W\"urzburg,
  Germany}{}{0000-0001-5872-718X}{}

\authorrunning{Agrawal, Cabello, Kaufmann, Saurabh, Sharma, Uno, and Wolff}

\Copyright{Akanksha Agrawal, Sergio Cabello, Michael Kaufmann, Saket Saurabh, Roohani Sharma, Yushi Uno, Alexander Wolff}


\acknowledgements{We thank the organizers of the 2023 Dagstuhl Seminar
  ``New Frontiers of Parameterized Complexity in Graph Drawing'',
  where this work was initiated.}

\funding{Funded in part by Science and Engineering Research Board, Startup Research Grant (SRG/2022/000962). 
Funded in part by the Slovenian Research and Innovation Agency (P1-0297, J1-2452, N1-0218, N1-0285).
Funded in part by the EU (ERC, KARST, project no.\ 101071836). Views and opinions expressed are however those of the authors only and do not necessarily reflect those of the EU or the ERC. Neither the EU nor the granting authority can be held responsible for them. 
Partially supported by
JSPS KAKENHI grant no.\ JP17K00017, 20H05964, and 21K11757.}

\newcommand{\OO}{\mathcal{O}}
\newcommand{\yes}{\textsc{Yes}}

\newcommand{\arxiv}[2]{#1}

\newcommand{\delpplanar}{\textsc{Edge Deletion to 1-Page \pl-Planar}\xspace}
\newcommand{\delppageplanar}[1]{\textsc{Edge Deletion to $#1$-Page
    Planar}\xspace}
\newcommand{\deldegp}{\textsc{Vertex Deletion to Degree-\pl}\xspace}
\newcommand{\spn}{\textnormal{\texttt{span}}}
\newcommand{\lef}{\textnormal{\texttt{left}}}
\newcommand{\cros}{\textnormal{\texttt{cross}}}
\newcommand{\opt}{\operatorname{OPT}}
\newcommand{\CR}{\operatorname{cr}}
\newcommand{\pg}{\ensuremath{p}\xspace}
\newcommand{\tr}{\ensuremath{t}\xspace}
\newcommand{\pl}{\ensuremath{d}\xspace}
\newcommand{\tw}{\operatorname{tw}}
\newcommand{\pw}{\operatorname{pw}}
\newcommand{\vc}{\operatorname{vc}}
\newcommand{\N}{\ensuremath{\mathcal{N}}\xspace}

\definecolor{defblue}{rgb}{0.1,0.4,0.6} %
\let\emph\relax\DeclareTextFontCommand{\emph}{\color{defblue}\em}

\graphicspath{{figures/}}

\usepackage{framed}
\usepackage{tabularx}
\usetikzlibrary{calc}

\newlength{\RoundedBoxWidth}
\newsavebox{\GrayRoundedBox}
\newenvironment{GrayBox}[1]%
   {\setlength{\RoundedBoxWidth}{.93\columnwidth}
    \def\boxheading{#1}
    \begin{lrbox}{\GrayRoundedBox}
       \begin{minipage}{\RoundedBoxWidth}}%
   {   \end{minipage}
    \end{lrbox}
    \begin{center}
    \begin{tikzpicture}%
       \node(Text)[draw=black!20,fill=white,rounded corners,inner sep=2ex,text width=\RoundedBoxWidth]
             {\usebox{\GrayRoundedBox}};
        \coordinate(x) at (current bounding box.north west);
        \node [draw=white,rectangle,inner sep=3pt,anchor=north west,fill=white]
        at ($(x)+(6pt,.75em)$) {\boxheading};
    \end{tikzpicture}
    \end{center}}

\newenvironment{defproblemx}[1]{\noindent\ignorespaces%
                                \FrameSep=6pt%
                                \parindent=0pt%
                \vspace*{-.5em}
                \begin{GrayBox}{#1}%
                \begin{tabular*}{\columnwidth}{!{\extracolsep{\fill}}@{\hspace{.1em}} >{\itshape} p{1.5cm} p{0.86\columnwidth} @{}}%
            }{
                \end{tabular*}%
                \end{GrayBox}%
                \ignorespacesafterend
                \vspace*{-.5em}
            }

\newcommand{\parProblemQuestion}[4]{%
  \begin{defproblemx}{#1}
    Input: & #2 \\
    Parameters: & #3 \\
    Question: & #4 \\[-1ex]
  \end{defproblemx}
}

\nolinenumbers %

\arxiv{\hideLIPIcs}{
\relatedversion{}
\relatedversiondetails{Full Version}{https://arxiv.org/abs/2403.XXXXX}
}


\EventEditors{Hans Bodlaender}
\EventNoEds{1}
\EventLongTitle{19th Scandinavian Symposium on Algorithm Theory (SWAT 2024)}
\EventShortTitle{SWAT 2024}
\EventAcronym{SWAT}
\EventYear{2024}
\EventDate{June 12-14, 2024}
\EventLocation{Helsinki, Finland}
\SeriesVolume{YYY}
\ArticleNo{XX}     %

\begin{document}

\maketitle

\begin{abstract}
  Drawing a graph in the plane with as few crossings as possible is
  one of the central problems in graph drawing and computational
  geometry.  Another option is to remove the smallest number of vertices
  or edges such that the remaining graph can be drawn without
  crossings.  We study both problems in a book-embedding setting for
  \emph{ordered graphs}, that is, graphs with a fixed vertex order.
  In this setting, the vertices lie on a straight line, called the
  \emph{spine}, in the given order, and each edge must be drawn on one
  of several pages of a book such that every edge has at most a fixed
  number of crossings.  In book embeddings, there is another way to
  reduce or avoid crossings; namely by using more pages.  The minimum
  number of pages needed to draw an ordered graph without any
  crossings is its (fixed-vertex-order) \emph{page number}.

  We show that the page number of an ordered graph with $n$ vertices
  and $m$ edges can be computed in $2^m \cdot n^{\OO(1)}$ time.  An
  $\OO(\log n)$-approximation of this number can be computed
  efficiently.  We can decide in
  $2^{\OO(d \sqrt{k} \log (d+k))} \cdot n^{\OO(1)}$ time whether it
  suffices to delete $k$ edges of an ordered graph to obtain a
  \emph{$d$-planar} layout (where every edge crosses at most $d$ other
  edges) on {\em one} page. As an additional parameter, we consider
  the size $h$ of a \emph{hitting set}, that is, a set of points on
  the spine such that every edge, seen as an open interval, contains
  at least one of the points.  For $h=1$, we can efficiently compute
  the minimum number of edges whose deletion yields fixed-vertex-order
  page number~$p$.  For $h>1$, we give an XP algorithm with respect to
  $h+p$.  Finally, we consider \emph{spine+\tr-track drawings}, where
  some but not all vertices lie on the spine.  The vertex order on the
  spine is given; we must map every vertex that does not lie on the
  spine to one of \tr tracks, each of which is a straight line on a
  separate page, parallel to the spine.  In this setting, we can
  minimize in $2^n \cdot n^{\OO(1)}$ time either the number of
  crossings or, if we disallow crossings, the number of tracks.
\end{abstract}

\section{Introduction}
\label{sec:intro}

Many crossings typically make it hard to understand the drawing of a graph, and thus much effort in the area of Graph Drawing has been directed towards reducing the number of crossings in drawings of graphs.
In terms of parameterized complexity, several facets of this problem have been considered.
For example, there are FPT algorithms that, given a graph~$G$ and an integer~$k$, decide whether $G$ can be drawn with at most $k$ crossings \cite{g-ccnqt-JCSS04,kr-ccnlt-STOC07}.
Crossing minimization has also been considered in the setting where each vertex of the given graph must lie on one of two horizontal lines.
This restricted version of crossing minimization is an important subproblem in drawing layered graphs according to the so-called \emph{Sugiyama framework}~\cite{stt-mvuhss-TSMC81}. 
There are two variants of the problem; either the vertices on both lines may be freely permuted or the order of the vertices on one line is given.
These variants are called
two-layer and one-layer crossing minimization, respectively.
For both, FPT algorithms exist \cite{kt-fssfp-Algorithmica15,kt-ffpt2-IPL16}.  Zehavi \cite{z-pacmp-CSR22} has surveyed parameterized approaches to crossing minimization. 

Surprisingly, crossing minimization remains NP-hard even when restricted to graphs that have a planar subgraph with just one edge less \cite{cm-a1e-SICOMP13}.
Another way to deal with crossings is to remove a small number of vertices or edges such that the remaining graph can be drawn without crossings.
In fact, it is known that vertex deletion to planarity is FPT with respect to the number of deleted vertices~\cite{JansenLS14,Kawarabayashi09,MarxS12}.
However, the running times of these algorithms depends at least
exponentially on the number of deleted vertices. On the kernelization
front, there exists an $\OO(1)$-approximate kernel for vertex deletion to
planarity~\cite{DBLP:conf/stoc/Jansen22}, whereas vertex deletion to
outerplanarity is known to admit an (exact) polynomial
kernel~\cite{DonkersJW22}.

In this paper, we focus on another model to cope with the problem of crossing edges, 
namely \emph{book embeddings}, 
drawings where the vertices lie on a straight line, called the \emph{spine},  
and each edge must be drawn on one of several halfplanes, called \emph{pages}, 
such that the drawing on each page is crossing-free (planar) or such that each edge has at most a constant number $c$ of crossings (that is, the drawing is $c$-planar). 
We consider the variant of the problem where the order $\sigma$ of the vertices is given and fixed.
The minimum number of pages to draw an (ordered) graph without any crossings 
is its (fixed-vertex-order) \emph{page number}.

In this paper, we study the problem of designing parameterized algorithms, 
where the possible parameters are the number $k$ of edges to be deleted, the number $c$ of allowed crossings per edge, the number $\pg$ of pages, and their combinations. 

\subparagraph*{Problem description.}

Given a graph $G$, let \emph{$V(G)$} denote the vertex set and \emph{$E(G)$} the edge set of~$G$.  
An \emph{ordered graph} $(G,\sigma)$ consists of a graph $G$ and an
ordering of the vertices of $G$, that is, a bijective map $\sigma\colon V(G) \to \{1,\dots, |V(G)|\}$.  
Henceforth, we specify every edge $(u,v)$ of $(G,\sigma)$ such that $\sigma(u)<\sigma(v)$.
For two edges $e=(u,v)$ and $e'=(u',v')$ of an ordered graph
$(G,\sigma)$, we say that $e$ and $e'$ \emph{cross with respect to
$\sigma$} if their endpoints interleave, that is, if $\sigma(u) < \sigma(u') < \sigma(v) < \sigma(v')$ or if
$\sigma(u') < \sigma(u) < \sigma(v') < \sigma(v)$.  
The ordered graph models the scenario
where the vertices of $G$ are placed along a horizontal line in the given order $\sigma$ and all the edges are drawn above the line
using curves that cross as few times as possible.  Whenever $e$ and
$e'$ cross with respect to $\sigma$, their curves must
intersect. Whenever $e$ and $e'$ do \emph{not} cross with respect to
$\sigma$, their curves can be drawn without intersections; for
example, we may use halfcircles.  In this setting, we get a drawing
such that two edges of $G$ cross precisely if and only if they cross
with respect to $\sigma$.
Given a positive integer $\pl$, we say that an ordered graph
$(G,\sigma)$ is \emph{$\pl$-planar} if every edge in~$G$ is crossed by
at most~$\pl$ other edges (where 0-planar simply means planar).

In this paper, we focus on fast parameterized algorithms for 
the following problem.
\parProblemQuestion{\sc Edge Deletion to \pg-Page \pl-Planar}{%
  An ordered graph $(G,\sigma)$ and positive integers $k$, $\pg$, and
  $\pl$.}{$k$, \pg, \pl}{Does there exist a set $S$ of at most $k$
  edges of $G$ such that $(G-S,\sigma)$ is \pg-page \pl-planar?}

\noindent
We stress that we view $\pg$ and $\pl$, though they appear in the problem name, not as constants, but as parameters.

\subparagraph*{Related work.}

Given an ordered graph~$(G,\sigma)$, its \emph{conflict graph}
$H_{(G,\sigma)}$ is the graph that has a vertex for each edge of~$G$
and an edge for each pair of crossing edges of~$G$.
Note that~$H_{(G,\sigma)}$ is a \emph{circle graph}, that is, 
the intersection graph of chords of a circle, because two chords in a circle
intersect if and only if their endpoints interleave.

We can express \delpplanar as the problem of deleting from~$H_{(G,\sigma)}$ a set of at most $k$ vertices such that the remaining graph has maximum degree at most~$\pl$.
For general graphs, this problem is called \deldegp~\cite{NishimuraRT05}; it admits a quadratic kernel~\cite{FellowsGMN11,Xiao17}.

Testing whether $(G,\sigma)$ has (fixed-vertex-order) page number
$\pg$ (without any edge deletions) is equivalent to the 
$\pg$-colorability of the conflict graph $H_{(G,\sigma)}$. 
For $\pg=2$, it suffices to test whether the conflict graph $H_{(G,\sigma)}$
is bipartite. An alternative approach, discussed by Masuda, Nakajima,
Kashiwabara, and Fujisawa~\cite{MasudaNKF90}, is to add to~$G$
a cycle connecting the vertices along the spine in the given order, 
and then test for planarity.
Another possibility is to use \textsc{2-Sat}.
For $p=4$, Unger~\cite{Unger88} showed that the problem is NP-hard. 
For $p=3$, he~\cite{Unger92} claimed an efficient solution, but
recently his approach was shown to be incomplete~\cite{Bachmann23}.

\delppageplanar{\pg} is the special case where %
$c=0$; it can be interpreted as deletion of as few
vertices as possible in the conflict graph $H_{(G,\sigma)}$ to obtain 
a $\pg$-colorable graph.
For $\pg=1$, the problem can be solved by finding a maximum independent set in a circle graph, which takes linear time \cite{Gavril73,Nash10,Valiente03}; see \cref{lem:fopn_iteration} in \cref{sec:page-number}.
\delppageplanar{2} can be phrased as
{\sc Odd Cycle Transversal} in the conflict graph, which means that
it is FPT with respect to the number of edges that must 
be deleted~\cite{ReedSV04}.
The case $\pg=2$ can also be modeled as a (geometric) special case of
\textsc{Almost 2-Sat (variable)}, which can be solved in
$2.3146^k \cdot n^{\OO(1)}$ time, where $k$ is the number of variables that
need to be deleted so that the formula becomes satisfiable
\cite[Corollary~5.2]{lokshtanov14}. 

Masuda et al.~\cite{MasudaNKF90} showed that the problem
\textsc{Fixed-Order 2-Page Crossing Number} is NP-hard.  In this problem, we have to decide, for each edge of the given ordered graph $(G,\sigma)$, whether to draw it above or below the spine, so as to minimize the number of crossings.  

Bhore, Ganian, Montecchiani, and N{\"{o}}llenburg~\cite{BhoreGMN20}studied
the fixed-vertex-order page number and 
provide an algorithm to compute it with running time $2^{\OO(\vc^3)}n$, where $\vc$ is the vertex cover number of the graph.  
They also proved that the problem is fixed-parameter tractable
parameterized by the pathwidth ($\pw$) of the {\em ordered} graph, with a
running time of $\pw^{\OO(\pw^2)}n$.
Note that the pathwidth of an ordered graph is in general not bounded by the vertex cover number~\cite{BhoreGMN20}.
This has been improved by Liu, Chen, Huang, and Wang~\cite{LiuCHW21} to $2^{\OO(\pw^2)}n$.
They also showed that the problem does not admit a polynomial kernel if parameterized only by~$\pw$ (unless NP $\subseteq$ coNP/poly).
Moreover, they gave an algorithm that checks in $(\CR + \,2)^{O(\pw^2)}n$ time whether a graph with $n$ vertices and pathwidth~$\pw$ can be drawn on a given number of pages with at most $\CR$ crossings in total.

Liu, Chen and Huang~\cite{LiuCH20} considered the problem \textsc{Fixed-Order Book Drawing} with bounded number of crossings per edge: decide if there is a $p$-page book-embedding of $G$ such that the maximum number of crossings per edge is upper-bounded by an integer $\pl$. 
This problem was posed by Bhore et al.~\cite{BhoreGMN20}. 
Liu et al.\ showed that this problem, when parameterized by both the maximum number~\pl of crossings per edge and the vertex cover number~$\vc$ of the graph, admits an algorithm running in $(\pl+2)^{\OO(\vc^3)}n$ time.
They also showed that the problem, when parameterized by both~\pl and the pathwidth~$\pw$ of the vertex ordering, admits an algorithm running in $(\pl+2)^{\OO(\pw^2)}n$ time.

All these problems can be considered also in the setting where we can choose the ordering of the vertices along the spine; see, for instance,~\cite{BhoreGMN20,chung1987embedding}.

\subparagraph*{Our contribution.}

For an overview over our results and known results, see \cref{tab:results}.
\begin{table}[tb]
  \caption{New and known results concerning \textsc{Edge Deletion to \pg-Page \pl-Planar}.}
  \label{tab:results}
  \centering\small
  \begin{tabular}{@{}cccccll@{}}
    \toprule
    $k$ & \pg & \pl & add. param. & ref. & \multicolumn{2}{l}{result (runtime, ratio, or kernel size)} \\ \midrule
    0 & $\min$ & 0 & -- & Cor.~\ref{cor:fopn_exact} & EXP: & $2^m n^{\OO(1)}$ \\
    0 & $\min$ & 0 & -- & Thm.~\ref{thm:fopn_approx} & approx: & ratio $\OO(\log n)$ \\
    0 & $\min$ & param. & -- & Cor.~\ref{cor:focppn-approx} & approx: & ratio $\OO((d+1)\log n)$ \\ \midrule
    param. & 1 & param. & -- & Thm.~\ref{thm:sub-exp} &  FPT: & $2^{\OO(\pl\sqrt{k} \log (\pl+k))} \cdot n^{\OO(1)}$ \\ \midrule
    $\min$ & param. & 0 & -- & Sect.~\ref{sec:del-P-page-planar} & EXP: & $4^m n^{\OO(1)}$ \\
    $\min$ & param. & 0 & $h=1$ & Thm.~\ref{thm:1-hitting-set} & P: & $\OO(m^3 \log n \log \log \pg)$ \\
    $\min$ & param. & 0 & $h$ & Thm.~\ref{thm:hitting-xp} & XP: & $\OO(m^{(4h-2)\pg+3} \log n \log \log \pg)$ \\ \midrule
    0 & -- & $\min$ & $t$ & Thm.~\ref{thm:k_tracks} & EXP: & $2^n n^{\OO(1)}$ \\
    0 & -- & $\min$ & $\min t$ & Cor.~\ref{cor:min-track} & EXP: & $2^n n^{\OO(1)}$ \\ \midrule
    param. & 1 & param. & -- & \cite{FellowsGMN11,Xiao17} & kernel: & quadratic \\
    0 & $\ge 4$ & 0 & -- & \cite{Unger88} & NPC. & \\
    0 & $\le 2$ & 0 & -- & folklore & P: & linear time; e.g., via \textsc{2-Sat} \\
    $\min$ & 1 & 0 & -- & e.g., \cite{Gavril73} & P: & linear time \\
    param. & 2 & 0 & -- & \cite{ReedSV04} & FPT: & \textsc{Odd Cycle Transversal} \\
    0 & 2 & $\min$ & -- & \cite{MasudaNKF90} & NPC. & \\
    0 & $\min$ & 0 & $\vc$ & \cite{LiuCH20} & FPT: & $(\pl+2)^{\OO(\vc^3)}n$ \\
    0 & $\min$ & 0 & $\pw$ & \cite{LiuCH20} & FPT: & $(\pl+2)^{\OO(\pw^2)}n$ \\
    0 & param. & $\CR$ & $\pw$ & \cite{LiuCHW21} & FPT: & $n \cdot (\CR + 2)^{\OO(\pw^2)}$ \\
    0 & param. & param. & $\pw$ & \cite{LiuCHW21} & FPT: & $2^{\OO(\pw^2)}n$; no poly.\ $\pw$-kernel \\
    \bottomrule
  \end{tabular}
\end{table}
First, we show that the fixed-vertex-order page number of an ordered
graph with $m$~edges and $n$~vertices can be computed in $2^m \cdot n^{\OO(1)}$ time; see
\cref{sec:page-number}.  
We use subset convolution~\cite{BjorklundHKK07}. 
Alternatively, given a budget $p$ of pages, we can compute a $p$-page
book embedding with the minimimum number of crossings. %
By combining the greedy algorithm
for \textsc{Set Cover} with an efficient algorithm for \textsc{Maximum
  Independent Set} in circle graphs \cite{Gavril73,Nash10,Valiente03},
we obtain an efficient $\OO((\pl+1) \log n)$-approxi\-ma\-tion
algorithm for the fixed-vertex-order \pl-planar page number.

Second, we tackle \delpplanar; see \cref{sec:sub-exp}.
We show how to decide in $2^{\OO(c\sqrt{k} \log (c+k))} \cdot n^{\OO(1)}$ time whether deleting $k$ edges of an ordered graph suffices to obtain a $c$-planar layout 
on {\em one} page.  Note that our algorithm is subexponential in~$k$.

Third, we consider the problem \delppageplanar{\pg}; see \cref{sec:del-P-page-planar}. 
As an additional parameter, we consider the size~$h$ 
of a \emph{hitting set}, that is, a set of points on the spine such 
that every edge, seen as an open interval, contains at least one of 
the points.
For $h=1$, we can efficiently compute the smallest set of edges whose deletion yields fixed-vertex-order
page number~$\pg$.  
For $h>1$, we give an XP algorithm with respect to $h+p$.

Finally, we consider \emph{spine+\tr-track 
drawings}; see \cref{sec:multiple-track}.
In such drawings, some but not all vertices lie on the spine.
The vertex order on the spine is 
again given, but now we must map every vertex that does not lie on 
the spine to one of \tr tracks, each of which is a straight line on a separate page, parallel to the spine.
Using subset convolution, we can minimize in $2^n \cdot n^{\OO(1)}$ time either 
the number of crossings or, if we disallow crossings,
the number of tracks.  

We close with some open problems; see \cref{sec:open}.

\section{Computing the Fixed-Vertex-Order Page Number}
\label{sec:page-number}

Let $(G,\sigma)$ be an ordered graph, 
and let \pg be a positive integer.
In this section, we consider \emph{$\pg$-page book-embeddings} of 
$(G,\sigma)$: the vertices of $G$ are placed on a \emph{spine} $\ell$
according to~$\sigma$, there are \pg \emph{pages} (halfplanes) sharing $\ell$ on their boundary, and for each edge we have
to decide on which page it is drawn.
The aim is to minimize the total number
of crossings for a given number of pages, or minimize the number of pages 
to attain no crossings; see 
\cref{fig:pages}.

\begin{figure}[b]
    \centering
	\includegraphics[page=4]{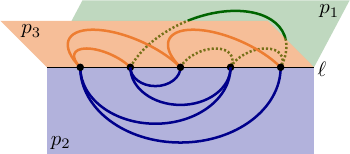}
	\caption{A 3-page book embedding of $K_5$ with fixed vertex order.
	For each edge, we can choose on which page it is drawn.  Note that
    $K_5$ cannot be drawn on two pages without crossings.
    }
	\label{fig:pages}
\end{figure}

Let $\CR_\pg(G,\sigma)$ be the minimum number of crossings
over all possible assignments of the edges of $E(G)$ to the \pg pages.
As discussed in the introduction, we can decide in linear
time whether $\CR_2(G,\sigma)=0$, but in general, 
computing $\CR_2(G,\sigma)$ is NP-hard~\cite{MasudaNKF90}.
The fixed-vertex-order page number of $(G,\sigma)$ is the minimum $\pg$ such that $\CR_\pg(G,\sigma)=0$.

\begin{theorem}
\label{thm:fopn_exact}
  Given an ordered graph~$(G,\sigma)$ with $n$~vertices and $m$~edges, 
  and a positive integer $\pg$, 
  we can compute the values $\CR_1(G,\sigma),\dots,\CR_\pg(G,\sigma)$ 
  in $2^m \cdot n^{\OO(1)}$ time.
  In particular, given a budget $p$ of pages, we can compute a $p$-page
  book embedding with the minimum number of crossings within the given time bound.
\end{theorem}
\begin{proof}
  Consider a fixed-vertex-order graph $((V,E),\sigma)$ with $n$
  vertices and $m$ edges.  We need to consider only the case $\pg < m$
  because, for $\pg \ge m$, it obviously holds that
  $\CR_{\pg}((V,E),\sigma)=0$.
  
  First note that, for any fixed $F\subseteq E$, 
  we can easily compute $\CR_{1}((V,F),\sigma)$ in 
  $\OO(|F|^2)=\OO(m^2)$ time by checking the order of the endpoints
  of each pair of edges.
  It follows that we can compute $\CR_{1}((V,F),\sigma)$ for all subsets 
  $F\subseteq E$ in $2^m \cdot n^{\OO(1)}$ time.
  
  For every $q>1$ and every $F\subseteq E$, we have the recurrence 
  \begin{equation}\label{eq:rec1}
	\CR_{q}((V,F),\sigma) ~=~ \min 
        \left\{ \CR_{1}((V,F'),\sigma) + \CR_{q-1}((V,F\setminus F'),\sigma) \mid 
        F'\subseteq F \right\}.
  \end{equation}
  Here, $F'\subseteq F$ corresponds to the edges that in the drawing
  go to one page, and thus $F\setminus F'$ goes to the remaining $q-1$ pages,
  where we can optimize over all choices of $F'\subseteq F$.
  
  From the recurrence in \cref{eq:rec1} we see that, for $q>1$, the function
  $F\mapsto \CR_{q}((V,F),\sigma)$ is, by definition, 
  the \emph{subset convolution} of the functions $F\mapsto \CR_{1}((V,F),\sigma)$ 
  and $F\mapsto \CR_{q-1}((V,F),\sigma)$ in the $(\min,+)$ ring. 
  Since $\CR_{q}((V,F),\sigma)$ takes integer values from $\{0,\dots,m^2 \}$ for
  every $q$ and $F$,
  it follows from~\cite{BjorklundHKK07} that one can obtain 
  $\CR_{q}((V,F),\sigma)$ for all $F\subseteq E$ in $2^m \cdot n^{\OO(1)}$ time,
  for a fixed $q>1$,
  assuming that $\CR_{1}((V,F),\sigma)$ and $\CR_{q-1}((V,F),\sigma)$ are already available. 
  Therefore, we can compute the values $\CR_{q}((V,F),\sigma)$ for $q \in \{2,\dots,\pg\}$ in
  $2^m \cdot n^{\OO(1)}$ time since $\pg \le m<n^2$.
\end{proof}

\begin{corollary}
\label{cor:fopn_exact}
  The fixed-vertex-order page number of a graph with 
  $n$~vertices and 
  $m$~edges can be computed in $2^m\cdot n^{\OO(1)}$ time. 
\end{corollary}

\begin{lemma}
\label{lem:fopn_iteration}
  Given an ordered graph $(G,\sigma)$, we can compute in polynomial time
  a smallest subset $S \subseteq E(G)$ such that $\CR_{1}(G-S,\sigma)=0$.
\end{lemma}

\begin{proof}
  Consider the \emph{conflict graph} $H_{(G,\sigma)}$ of $(G,\sigma)$,
  already defined in the Introduction.
  Note that $H_{(G,\sigma)}$ is a circle graph.
  Therefore, a largest independent set in $H_{(G,\sigma)}$
  corresponds to a largest subset $F$ of edges with $\CR_{1}((V,F),\sigma)=0$,
  which corresponds to a minimum set $S \subseteq E(G)$ such that
  $\CR_{1}(G-S,\sigma)=0$.
  Finally, note that a largest independent set in circle graphs 
  can be computed in polynomial time~\cite{Gavril73,Nash10,Valiente03}.
\end{proof}

\begin{theorem}
  \label{thm:fopn_approx}
  We can compute an $\OO(\log n)$-approximation to the fixed-vertex-order 
  page number of a graph with $n$ vertices in polynomial time. 
\end{theorem}

\begin{proof}
  Let $((V,E),\sigma)$ be the given ordered graph, and 
  let $\opt$ be its fixed-vertex-order page number. 
  Define the family $\mathcal{F}= \{ F\subseteq E\mid \CR_{1}((V,F),\sigma)=0\}$.
  Consider the \textsc{Set Cover} instance $(E,\mathcal{F})$, where $E$ is the universe
  and $\mathcal{F} \subseteq 2^E$ is a family of subsets of $E$.
  A feasible solution of this \textsc{Set Cover} instance is a subfamily 
  $\mathcal{F}'\subseteq \mathcal{F}$ such that $\bigcup\mathcal{F}' = E$.
  The task in \textsc{Set Cover} is to find a feasible solution of minimum cardinality.

  Each feasible solution $\mathcal{F}'$ to the \textsc{Set Cover} 
  instance $(E,\mathcal{F})$ corresponds to a fixed-vertex-order 
  drawing of $(G,\sigma)$ with $|\mathcal{F}'|$ pages. %
  Similarly, each fixed-vertex-order drawing of $(G,\sigma)$ 
  with $\pg$ pages represents a feasible solution to \textsc{Set Cover}
  with $\pg$ sets. %
  In particular, the size of the optimal solution to the {\sc Set Cover} 
  instance $(E,\mathcal{F})$ is equal to $\opt$, the fixed-vertex-order 
  page number of $(G,\sigma)$.
  
  Consider the usual greedy algorithm for \textsc{Set Cover}, which works as follows.
  Set $E_1=E$ and $i=1$.
  While $E_i\neq \emptyset$, we set $F_i$ to be the element of~$\mathcal{F}$
  that contains the largest number of edges from $E_i$,
  increase $i$, and set $E_i=E_{i-1}\setminus F_{i-1}$.
  Let $i^\star$ be the maximum value of $i$ with $E_i \ne \emptyset$.  Thus
  $E_{i^\star+1}=\emptyset$, and the algorithm finishes.
  It is well known that $i^\star \le \opt \cdot \log |E|$; 
  see for example~\cite[Section~5.4]{DasguptaPV-book}.
  Therefore, this greedy algorithm yields an
  $\OO(\log |V|)$-approximation for our problem.
  
  Finally, note that the greedy algorithm can be implemented to run efficiently. 
  Indeed, $F_i$ can be computed from $E_i$
  in polynomial time because of \cref{lem:fopn_iteration},
  and the remaining computations in every iteration are trivially done 
  in polynomial time. The number of iterations is polynomial because $i^\star \le |E|$.
\end{proof}

\begin{corollary}
  \label{cor:focppn-approx}
  We can compute an $\OO((\pl+1) \log n)$-approximation to the
  fixed-vertex-order $\pl$-planar page number of a graph with $n$
  vertices in polynomial time.
\end{corollary}

\begin{proof}
  Consider first an ordered graph $(H,\sigma)$ that is $\pl$-planar if
  drawn on a single page, with $\pl>0$.  Let $F_\pl$ be the subset
  of~$E(H)$ such that each edge in $F_c$ participates in exactly \pl
  crossings, and let $S_\pl$ be a maximal subset of $F_\pl$ such that
  no two edges in~$S_\pl$ cross each other.  Then, $(H-S_\pl,\sigma)$
  is $(\pl-1)$-planar because each edge of $H$ has fewer than $\pl$
  crossings, is in $S_c$, or is crossed by some edge in~$S_\pl$.  It
  follows by induction that $(H,\sigma)$ can be embedded in $\pl+1$
  pages without crossings.

  Consider now the input ordered graph $(G,\sigma)$ and
  let $\opt_\pl$ be the minimum number of $\pl$-planar pages
  needed for $(G,\sigma)$.  By the argument before applied to each page, 
  we know that the minimum number of planar pages, $\opt_0$, is 
  at most $(\pl+1) \opt_\pl$.  Using \cref{thm:fopn_approx}, 
  we obtain a drawing of $(G,\sigma)$ without crossings with at most 
  $\opt_0 \cdot \, \OO(\log n) \le (\pl+1)\opt_\pl \cdot \, \OO(\log n)$ (planar)
  pages, where $n=|V(G)|$.  Such a drawing is of course also $\pl$-planar.
\end{proof}

\section{Edge Deletion to 1-Page \boldmath{\pl}-Planar}
\label{sec:sub-exp}

The main result of this section is as follows.

\begin{theorem}\label{thm:sub-exp}
  \delpplanar\ admits an algorithm with running time
  $2^{\OO(\pl\sqrt{k} \log (\pl+k))} \cdot n^{\OO(1)}$, where $n$ is
  the number of vertices in the input graph and $k$ is the number of
  edges to be deleted.
\end{theorem}

In other words, we obtain a subexponential fixed-parameter tractable
algorithm for \delpplanar\ parameterized by $k$, the number of edges to be deleted; note
that we consider~$\pl$ to be a constant here (although we made explicit
how the running time depends on~$\pl$).  Our algorithm to prove
Theorem~\ref{thm:sub-exp} has two steps.  First it branches on edges
that are crossed by at least $\pl+\sqrt{k}$ other edges.  
When such edges do not exist, we show that the conflict graph $H_{(G,\sigma)}$
has treewidth $\OO(\pl+\sqrt{k})$. This is done by showing that the conflict graph has balanced separators. Finally the bound on the treewidth
allows us to use a known
(folklore) algorithm~\cite{DBLP:journals/corr/abs-2304-14724} for
\deldegp whose dependency is singly exponential in the treewidth of
$H_{(G,\sigma)}$.

\subsection{Branching}
\label{sec:branching}

Let $\cros_{(G,\sigma)}(e)$ denote the set of edges of~$G$ that cross~$e$ with respect to~$\sigma$.  
We drop the subscript $(G,\sigma)$ when it is clear from the context.  
We show that we can use branching to reduce any instance to a
collection of instances where each edge $e$ of the graph satisfies
$|\cros(e)| < \pl+\sqrt{k}$. In particular we show the following lemma.

\begin{lemma}\label{lem:branching}
  Let $(G,\sigma,k)$ be an instance of \delpplanar.  There is a
  $2^{\OO(\pl \sqrt{k} \log(\pl+k))} \cdot n^{\OO(1)}$-time algorithm that
  outputs $2^{\OO(\pl \sqrt{k} \log(\pl+k))}$ many instances of
  \delpplanar\ $(G_1,\sigma,k_1), \ldots, (G_r, \sigma,k_r)$ such that
  for each $i\in [r]$, $G_i$ is a $(\pl+\sqrt k)$-planar graph, and
  $(G,\sigma,k)$ is a \yes-instance of \delpplanar\ if and only if
  $(G_i,\sigma,k_i)$ is a \yes-instance of \delpplanar\ for some
  $i \in [r]$.
\end{lemma}

\begin{proof}
  Let $e$ be an edge of~$G$ with
  $|\cros(e)| \ge \pl+\lceil\sqrt{k}\rceil$.  If $|\cros(e)| > \pl+k$,
  then $e$ must be deleted, as we cannot afford to keep $e$ and delete
  enough edges from $\cros(e)$.  If $|\cros(e)| \le \pl+k$, then either
  $e$ must be deleted or at least $|\cros(e)|-\pl$ many edges from
  $\cros(e)$ must be deleted, so that at most $\pl$ edges of $\cros(e)$
  stay.  This results in the following branching rule, where we return
  an OR over the answers of the following instances:
  \begin{enumerate}
  \item Recursively solve the instance $(G-e,\sigma, k-1)$. This
    branch is called the \emph{light} branch.
  \item If $|\cros(e)|> \pl +k$, we do not consider other branches.
    Otherwise, for each subset $X$ of $\cros(e)$ with $|\cros(e)|-\pl$
    many edges, recursively solve the instance $(G-X,\sigma,
    k-|X|)$. Each of these branches is called a \emph{heavy} branch. 
  \end{enumerate}
  We are going to show that the recursion tree has
  $2^{\OO(\pl\sqrt{k} \log (\pl+k))}$ branches.  
  Note that the number of possible heavy branches at each is node is
  \[
    \binom{|\cros(e)|}{|\cros(e)|-\pl} ~=~ \binom{|\cros(e)|}{\pl} ~\le~
    \binom{\pl+k}{\pl} ~\le~ (\pl+k)^\pl.
  \]
  To prove the desired upper bound, we interpret the branching tree as
  follows. First note that, in each node, we have at most $(\pl+k)^\pl$
  heavy branches.  We associate a distinct word over the alphabet
  $\Sigma=\{0, 1,\dots, (\pl+k)^\pl\}$ to each leaf (or equivalently each root to leaf path) 
  of the recurrence tree.  For each node of the recurrence tree,
  associate a character from $\Sigma$ with each of its children such
  that the child node corresponding to the light branch gets 
  the character $0$ and the other nodes (corresponding to the heavy
  branches) get a distinct character from $\Sigma \setminus \{0\}$.  
  Now a word over
  the alphabet $\Sigma$ for a leaf~$\ell$ of the recurrence tree is
  obtained by taking the sequence of characters on the nodes of the
  root to leaf $\ell$ path in order.  
  In order to bound the number
  of leaves (and hence the total number of nodes) of the recurrence
  tree, it is enough to bound the number of such words.
  The character~$0$ is called a \emph{light} label and all other
  characters are called \emph{heavy} labels.  Recall that a light
  label corresponds to the branch where $k$ drops by $1$, while the
  heavy labels correspond to the branches where $k$ drops by
  $|\cros(e)|-\pl \ge \sqrt{k}$.  This implies that each word (that is
  associated with the leaf of the recurrence tree) has at most
  $\sqrt{k}$ heavy labels. In order to bound the number of such words,
  we first guess the places in the word that are occupied by heavy
  labels and then we guess the (heavy) labels themselves at these
  selected places. All other positions have the light label on them
  and there is no choice left. Hence, the number of such words is
  upper-bounded by
  \[\sum_{i=0}^{\sqrt{k}} {k \choose i} ((\pl+k)^\pl)^{i} ~\leq~ 
     \sqrt{k} {k \choose
     \sqrt{k}} ((\pl+k)^\pl)^{\sqrt k} ~=~ 2^{\OO(\pl\sqrt{k} \log
     (\pl+k))}.\]

 This shows that the number of such words is bounded by
 $2^{\OO(\pl\sqrt{k} \log (\pl+k))}$, and hence the number of leaves (and
 nodes) of the recurrence tree is bounded by
 $2^{\OO(\pl\sqrt{k} \log (\pl+k))}$.
\end{proof}

\subsection{Balanced Separators in the Conflict Graph}
\label{sec:balanced-separators}

Let $(G,\sigma)$ be an ordered graph.  For any edge $e=(u,v)$ of~$G$,
let $\spn_{(G,\sigma)}(e)$ be the set of all edges $(u',v') \ne e$ of
$G$ such that $\sigma(u) \le \sigma(u') \le \sigma(v') \le \sigma(v)$.
For example, in Figure~\ref{fig:case1}, $\spn(e)=\{e_1\}$.
For any vertex~$w$ of~$G$, let $\lef_{(G,\sigma)}(w)$ be the set of
all edges $(u,v)$ of $G$ such that
$\sigma(u) < \sigma(v) \le \sigma(w)$.  Whenever it is clear from the
context, we will drop the subscript $(G,\sigma)$.  We say that an edge
$e$ of~$G$ is \emph{maximal} if $G$ contains no edge $e'$ such that
$e \in \spn(e')$,

\begin{lemma}[Balanced Separator in the Conflict Graph]
  \label{lem:separator}
  If $(G,\sigma)$ is an ordered $\pl$-planar graph, then $G$ contains
  a set $X$ of at most $3(\pl+1)$ edges such that
  $E(G) \setminus X = E_1 \cup E_2$, $E_1 \cap E_2 = \emptyset$,
  $|E_1| \le 2m/3$, $|E_2| \le 2m/3$, and no edge $e_1 \in E_1$
  crosses an edge $e_2 \in E_2$ with respect to~$\sigma$.
\end{lemma}

\arxiv{
\begin{proof}
  We consider three cases depending on the spans of the edges of~$G$.

  \smallskip

  \noindent\textbf{Case 1:} There exists an edge $e=(u,v) \in E(G)$ such that
  $m/3 \le |\spn(e)| \le 2m/3$.

  \smallskip
    
  In this case, let $X = \cros(e) \cup \{e\}$, let $E_1 = \spn(e)$,
  and let $E_2 = E(G) \setminus (X\cup E_1)$.  Note that, by
  construction, $|X| \leq \pl +1$,
  $|E_1| \le 2m/3$, and $|E_2| \le 2m/3$.  Now let
  $e_1=(u_1,v_1) \in E_1$ and $e_2=(u_2,v_2) \in E_2$.  Since
  $e_1 \in E_1$, we have
  $\sigma(u) \le \sigma(u_1) < \sigma(u_2) \le \sigma(v)$; see
  \cref{fig:case1}.  Since $e_2 \in E_2$, we have
  $\sigma(v_2)\le \sigma(u)$ or $\sigma(v)\le\sigma(u_2)$; see the black
  and the gray versions of~$e_2$ in \cref{fig:case1}, respectively.
  In both cases, $e_1$ and $e_2$ do not cross.

  \smallskip
    
  \noindent\textbf{Case 2:} For every edge $e \in E(G)$, it holds
  that $|\spn(e)| \leq m/3$.

  \smallskip

  Let $M \subseteq E(G)$ be the collection of all maximal edges of $G$
  in $\sigma$.  Let $\mu=|M|$, and let
  $M=\{(u_1,v_1), \dots, (u_\mu,v_\mu)\}$, where
  $\sigma(u_1) < \sigma(u_2) < \dots < \sigma(u_\mu)$.  Note that
  $|\lef(u_1)| \le \dots \le |\lef(u_\mu)|$ and that $|\lef(u_1)|=0$.
  The equality is due to the fact that $u_1$ is the first non-isolated
  vertex of~$G$ in~$\sigma$ (and $v_1$ is the rightmost neighbor
  of~$u_1$).
    
  Let $a \in [\mu]$ 
  be the largest index such that
  $|\lef(u_a)| \leq m/3$.  Since $|\lef(u_1)|=0$, it is clear that
  such an index $a$ exists. Moreover, we have $a<\mu$. This is because 
  $|\lef(v_{a})| \le |\lef(u_a)| + |\spn(u_a,v_a)|\le 2m/3$ and $\lef(v_{\mu})=E(G)$.
  Therefore, $a+1\in [\mu]$.

  We claim that $m/3 < |\lef(u_{a+1})| \leq 2m/3 + \pl+1$.  From the
  choice of $a$, it is clear that $|\lef(u_{a+1})| > m/3$.  Note that
  $\lef(u_{a+1}) \subseteq \lef(u_a) \cup \spn((u_a,v_a)) \cup
  \cros((u_a,v_a))\cup \{ (u_a,v_a)\}$; see \cref{fig:case2}.  
  This yields our claim since
  $|\lef(u_{a+1})| \leq |\lef(u_a)| + |\spn((u_a,v_a))| +
  |\cros((u_a,v_a))| +1 \leq 2m/3 + \pl+1$.

  Now let $X = \cros((u_{a+1},v_{a+1})) \cup \{(u_{a+1},v_{a+1})\}$,
  $E_1 = \lef(u_{a+1})$ and $E_2= E(G)\setminus (X\cup E_1)$.  Since
  $m/3 \le |\lef(u_{a+1})| \le 2m/3+\pl+1$, $|E_1|\le 2m/3+\pl+1$, and
  $|E_2| \le 2m/3$.  Finally, we simply move $\pl+1$ edges from $E_1$ to
  $X$.  Then $|X| \le 2(\pl+1)$ and $|E_1|\le 2m/3$.  Given our
  construction, it is clear that no edge in~$E_1$ crosses any edge
  in~$E_2$; see \cref{fig:case2}.

  \smallskip
    
  \noindent\textbf{Case 3:} There exists an edge $e \in E(G)$ such
  that $|\spn(e)| > 2m/3$.

  \smallskip

  \noindent 
  Let $e=(u,v)$ be an edge of~$G$ such that $|\spn(e)| > 2m/3$ and
  there is no $e' \in \spn(e)$ such that $|\spn(e')| > 2m/3$.  Let
  $V'=\{w \in V(G) \colon \sigma(u) \le \sigma(w) \le
  \sigma(v)\}$.  Let $G'=G[V']$, and let $\sigma'$ be the
  restriction of~$\sigma$ to~$V'$.

  Since Case~1 does not apply, for each
  $e' \in \spn(e)$, we have $|\spn(e')| \leq m/3$.  Therefore, Case~2
  applies to the ordered graph $(G',\sigma')$.  This yields a
  set $X' \subseteq E(G')$ of size at most $2(\pl+1)$,
  and disjoint sets $E'_1$ and $E'_2$ of edges such that
  $E(G') \setminus X' = E'_1 \cup E'_2$, 
  $m/3 \le |E'_1| \le 2m/3$, $|E'_2| \le 2m/3$, and no edge in
  $E'_1$ crosses any edge in $E'_2$.

  Let $X=X' \cup \cros(e) \cup \{e\}$.  Then $|X| \le 3(\pl+1)$.
  Let $E_1 = E'_1$ and $E_2 = E(G)\setminus (X\cup E_1)$.  Since
  $m/3 \le |E'_1| \le 2m/3$, clearly $|E_2| \le 2m/3$.  It remains
  to show that no edge of $E_2$ crosses any edge of $E_1$; see
  \cref{fig:case3}.  By construction, no edge of $E'_2$ crosses
  any edge of $E'_1$.  The edges in $E_2 \setminus E'_2$
  neither cross~$e$ nor do they lie in $\spn(e)$, so they cannot cross
  any edge in~$E_1 \subseteq \spn(e)$.
\end{proof}}{
\begin{proof}[Proof sketch]
  The proof is in the full version.  We consider three cases depending on the spans of the edges of~$G$;
  see \cref{fig:separator}.  Either there exists an edge
  $e=(u,v) \in E(G)$ such that $m/3 \le |\spn(e)| \le 2m/3$ (case~1),
  or for every edge $e \in E(G)$, it holds that $|\spn(e)| \leq m/3$
  (case~2), or there exists an edge $e \in E(G)$ such that
  $|\spn(e)| > 2m/3$ (case~3). 
\end{proof}}

\begin{figure}[tb]
  \centering
  \begin{subfigure}[b]{.33\textwidth}
    \centering \includegraphics[page=1,scale=.9]{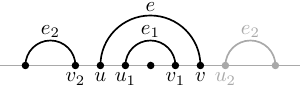}
    \caption{case 1}
    \label{fig:case1}
  \end{subfigure}
  \hfill
  \begin{subfigure}[b]{.23\textwidth}
    \centering \includegraphics[page=2,scale=.9]{separator}
    \caption{case 2}
    \label{fig:case2}
  \end{subfigure}
  \hfill
  \begin{subfigure}[b]{.4\textwidth}
    \centering \includegraphics[page=3,scale=.9]{separator}
    \caption{case 3}
    \label{fig:case3}
  \end{subfigure}
  \caption{Case distinction for the proof of \cref{lem:separator}.}
  \label{fig:separator}
\end{figure}

\subsection{Proof of Theorem~\ref{thm:sub-exp}}

We now need to establish a relation between the treewidth of the graph
and the size of a balanced separator in it. For this we use the result
of Dvo\v{r}{\'a}k and Norin~\cite{DBLP:journals/jct/DvorakN19} that
shows a linear dependence between the treewidth and the
\emph{separation number} of a graph: the separation number of a graph
is the smallest integer $s$ such that every subgraph of the given
graph has a balanced separator 
of size at most $s$.  
A balanced separator in a graph $H$ is a set of vertices $B$ such that the vertex set of $H-B$ can be partitioned into two parts $V_1$ and $V_2$ such that $E(V_1,V_2) =\emptyset$ and $|V_1|, |V_2| \leq 2|V(H)|/3$.
In other words,
they show that if the separation number of the graph is $s$, then the
treewidth of such a graph is $\mathcal{O}(s)$.

Recall that $(G,\sigma,k)$ is an instance of \delpplanar.  By
Lemma~\ref{lem:separator}, if the ordered graph $(G,\sigma)$ is
$(\pl+\sqrt{k})$-planar, then the conflict graph $H_{(G,\sigma)}$ has a
balanced separator of size at most $3(\pl+\sqrt{k}+1)$.  Thus, due to
the result of Dvo\v{r}{\'a}k and
Norin~\cite{DBLP:journals/jct/DvorakN19}, the treewidth of
$H_{(G,\sigma)}$ is $\mathcal{O}(\pl+\sqrt{k})$.

Given a graph with $N$ vertices and treewidth~$\tw$, one can compute, in
$(d+2)^{\tw} \cdot N^{\OO(1)}$ time, the smallest set of vertices whose
deletion results in a graph of degree at most
$d$~\cite{DBLP:journals/corr/abs-2304-14724}.  Applying this
result to the conflict graph $H_{(G,\sigma)}$, which has at most
$|V(G)|^2=n^2$ vertices and treewidth $\OO(\pl+\sqrt{k})$, we conclude
that \delpplanar\ can be solved in
$2^{\OO((\pl+\sqrt{k}) \log \pl)} \cdot n^{\OO(1)}$ time if the given
ordered graph $(G,\sigma)$ is $(\pl+\sqrt{k})$-planar.

From \cref{lem:branching}, we can assume, at the expense of a
multiplicative factor of
$2^{\OO(\pl\sqrt{k} \log (k+\pl))} \cdot n^{\OO(1) }$ on the running time,
that the given ordered graphs $(G,\sigma)$ to consider are $(\pl+\sqrt{k})$-planar.
Thus, given $(G,\sigma,k)$, we can solve \delpplanar\ in
$2^{\OO(\pl\sqrt{k} \log (\pl+k))} \cdot n^{\OO(1) }$ time.  This
concludes the proof of \cref{thm:sub-exp}.

\section{Edge Deletion to \boldmath{\pg}-Page Planar}
\label{sec:del-P-page-planar}

In this section we treat the problem \delppageplanar{\pg}, which is the special case of \textsc{Edge Deletion to \pg-Page \pl-Planar} for $\pl=0$.
It can be solved by brute force in $\OO((\pg+1)^m \cdot n^2)$ time:
For each mapping of the $m$ edges to the $\pg$ pages, 
with the ``+1'' to mark edge deletion, 
check for each pair of edges assigned to the same page whether they intersect.
It can also be solved in $4^m\cdot n^{\OO(1)}$ time: for each of
the $2^m$ subsets of $E(G)$, use \cref{cor:fopn_exact} to 
decide whether its fixed-vertex-order page number is at most $\pg$.

We now consider a new parameter in addition to~\pg.  The edge set of
an ordered graph $(G,\sigma)$ corresponds to a set of open intervals
on the real line; namely every edge $(u,v)$ of~$G$ is mapped to the
interval $(\sigma(u),\sigma(v))$.  Given a set~$\mathcal{I}$ of
intervals, a \emph{hitting set} for~$\mathcal{I}$ is a set of points
on the real line such that each interval contains at least one of the
points.  Note that %
a hitting set
can be much smaller than a vertex cover: an ordered
graph $(G,\sigma)$ with a hitting set of size~1 can have linear vertex
cover number (e.g., $G=K_{n,n}$). \arxiv{

}{}Given a set~$\mathcal{I}$ of $m$ open intervals, a
minimum-size hitting set for~$\mathcal{I}$ can be found in
$\OO(m \log m)$ time by \arxiv{the following}{a} simple greedy
algorithm\arxiv{: sort
the intervals in~$\mathcal{I}$ by (non-decreasing) right endpoints,
then repeatedly put a point just before
the right endpoint~$v$ of the first interval~$e=(u,v)$ into the
hitting set under construction and delete from~$\mathcal{I}$ all
intervals (including~$e$) that contain~$v$.  Given an ordered graph
$(G,\sigma)$, let $h(G,\sigma)$ %
denote the minimum size of a hitting set for~$E(G)$}{}.

For two edges $(u,v),(u',v')$ of $(G,\sigma)$, 
we say that $(u,v)$ \emph{contains} $(u',v')$ if the interval
$(\sigma(u),\sigma(v))$ contains the interval $(\sigma(u'),\sigma(v'))$.
If $(u,v)$ and $(u',v')$ cross with respect to~$\sigma$, then there is 
no containment, otherwise one contains the other.

\subparagraph*{Hitting set of size~1.}

We start by treating the following special case of
\delppageplanar{\pg}.  Given an ordered graph~$(G,\sigma)$, a
point~$z$ on the real line that is contained in every interval defined
by~$E(G)$, a number $\pg$ of pages, and a threshold $k \ge 0$, we want
to decide whether there is a set $E' \subseteq E(G)$ of size at most
$k$ such that that $G-E'$ can be drawn without crossings on $\pg$
pages (respecting vertex order~$\sigma$).
Note that if there is a hitting set of size~1, then $G$ is necessarily
bipartite and that $z \notin \sigma(V(G))$.  We show that
\delppageplanar{\pg} can be solved efficiently if $h(G,\sigma)=1$.

Alam et al.~\cite{abgkp-mpng-TCS22} have called this setting
\emph{separated}; they showed that the \emph{mixed} page number of an
ordered $K_{n,n}$ is $\lceil 2n/3 \rceil$ in this case.  While we
study the (usual) page number of an ordered graph where each page
corresponds to a stack layout, the mixed page number asks for the
smallest number of stacks and queues (where nested edges are not
allowed on the same page) needed to draw an ordered graph.

\begin{theorem}
  \label{thm:1-hitting-set}
  Given an ordered graph $(G,\sigma)$ with $n$ vertices, $m$ edges,
  and $h(G,\sigma)=1$, \delppageplanar{\pg} can be solved in
  $\OO(m^3 \log n \log\log p)$ time.
\end{theorem}

\begin{proof}
  From $(G,\sigma)$ we construct an acyclic directed auxiliary
  graph~$G^+$,
  from which we then construct an $s$--$t$ flow network~\N
  such that an integral maximum $s$--$t$ flow of minimum cost in~\N
  corresponds to~$p$ vertex-disjoint directed paths in~$G^+$ of maximum
  total length, and each path in~$G^+$ corresponds to a set of edges
  in~$G$ that can be drawn without crossings on a single page in a
  book embedding of $(G,\sigma)$.  The set $E'$ of edges that need to
  be deleted from~$G$ such that $G-E'$ has page number~$p$ corresponds
  to the vertices of~$G^+$ that do not lie on any of the $p$ paths.

  We now describe these steps in detail.  The auxiliary graph~$G^+$
  has a node for each edge of~$G$ and an arc from edge
  node~$(a,b)$ to edge node~$(a',b')$ if 
  in $(G,\sigma)$ the edge $(a',b')$ contains the edge $(a,b)$
  (meaning that the edges do not cross);
  see \cref{fig:1-hitting-set}.  Hence $G^+$ has exactly $m$ nodes
  and at most $m \choose 2$ edges, and can be constructed
  from~$(G,\sigma)$ in $\OO(m^2)$ time.

  \begin{figure}[tb]
    \begin{subfigure}[b]{.42\textwidth}
      \centering
      \includegraphics{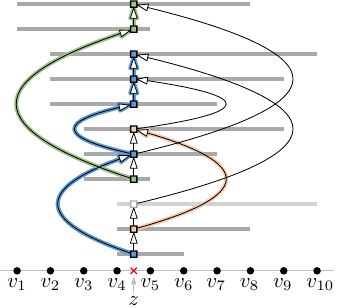}
      \caption{intervals corrsponding to the edges of~$G$; auxiliary
        graph $G^+$ without transitive edges}
      \label{fig:1-hitting-set-instance}
    \end{subfigure}
    \hfill
    \begin{subfigure}[b]{.54\textwidth}
      \hspace*{-3ex}
      \includegraphics{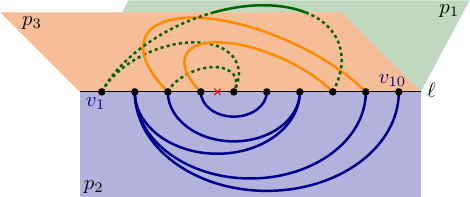}
      \caption{optimal solution for~(a): only the edge~$(v_4,v_{10})$ is
        deleted; the pages correspond to the colored paths in~(a)}
      \label{fig:1-hitting-set-solution}
    \end{subfigure}
    \caption{Instance with hitting set of size~1 and optimal solution
      for three pages.}
    \label{fig:1-hitting-set}
  \end{figure}

  The $s$--$t$ flow network~\N is defined as follows.  For each
  node~$v$ of~$G^+$, introduce two vertices~$v_\mathrm{in}$
  and~$v_\mathrm{out}$ in \N, 
  connected by the arc
  $(v_\mathrm{in},v_\mathrm{out})$ of capacity~1 and cost~$-1$. 
  All other arcs in~\N have cost~0.   
  For each arc~$(u,v)$ of~$G^+$, add the arc $(u_\mathrm{out},v_\mathrm{in})$ of capacity~1 to~\N.
  Then add to~\N new vertices~$s$, $s'$, and~$t$, the edge~$(s,s')$ of capacity~$p$, and the edges
  $\{ (s',v_\mathrm{in}), (v_\mathrm{out},t) \colon v \in V(G^+) \}$ of capacity~1.  
  Summing up, \N has $2m+3$ vertices, 
  at most ${m \choose 2} + 3m + 1$ edges, and can be constructed
  from~$G^+$ in $\OO(m^2)$ time.

  Due to the edge~$(s,s')$, a maximum flow in~\N has value at most~$p$.
  If $m \ge p$ (otherwise the instance is trivial, and no edge has to
  be deleted), then a maximum flow has value exactly~$p$.  Since all
  edge capacities and costs are integral, the minimum-cost circulation
  algorithm of Ahuja, Goldberg, Orlin, and
  Tarjan~\cite{agot-fmcfds-MP92} yields an integral flow.
  Since all edges (except for~$(s,s')$) have edge
  capacity~1 and \N~is acyclic, %
  the edges (except for~$(s,s')$) with 
  non-zero flow form $p$~paths of flow~1 from~$s'$ to~$t$ that are 
  vertex-disjoint except for their endpoints.  
  These paths (without $s'$ and~$t$) correspond to
  vertex-disjoint paths in~$G^+$.  Due to the negative cost of the edges
  of type $(v_\mathrm{in},v_\mathrm{out})$, the flow maximizes the number
  of such edges with flow.  This maximizes the number of vertices
  in~$G^+$ that lie on one of the $p$ paths.  This, in turn, maximizes
  the number of edges of $G$ that can be drawn without crossings on
  $p$ pages in a book embedding of $(G,\sigma)$.  Given a flow network
  with $n'$ vertices, $m'$ edges, maximum capacity~$U$, and maximum
  absolute cost value~$C$, the algorithm of Ahuja et al.\ runs in
  $\OO(n'm' (\log\log U) \log(n'C))$ time.  In our case,
  $n' \in \OO(m)$, $m' \in \OO(m^2)$, $U=p$, and $C=1$.  
  Hence computing the maximum flow of minimum cost in~\N takes
  $\OO(m^3 \log n \log\log p)$ time.  This dominates the time needed
  to construct~$G^+$ and~\N.
\end{proof}

In our forthcoming algorithm, we will use an extension of this result, 
as follows. Two subsets $E',E''\subset E(G)$ are \emph{compatible} if 
$|E'|=|E''|$ and there is an enumeration $e'_1,\dots,e'_{|E'|}$ of $E'$
and an enumeration $e''_1,\dots,e''_{|E'|}$ of $E''$
such that $e'_i$ is contained in $e''_i$ for all $i\in [|E'|]$.
Note that we may have $E'\cap E'' \ne \emptyset$.

\begin{lemma}
  \label{lem:1-hitting-set}
  Given an ordered graph $(G,\sigma)$ with $n$ vertices, $m$ edges,
  $h(G,\sigma)=1$, and subsets $E',E''\subset E(G)$ with $\pg=|E'|=|E''|$,
  we can decide, in $\OO(m^3 \log n \log\log p)$ time, whether $E'$ and $E''$ are compatible and, if yes,
  solve a version of \delppageplanar{\pg} where, on each page, one edge of $E'$
  is contained in all others edges and one edge of $E''$ contains all 
  other edges on that page.
\end{lemma}
\begin{proof}
	We adapt the proof of \cref{thm:1-hitting-set} by modifying the flow 
	network~\N that is considered. More precisely, we insert
	arcs from $s'$ only to the edges $e'\in E'$, and we insert arcs
	to $t$ only from the edgs $e''\in E''$. No other arcs go out from $s'$
	nor go into $t$. 
	
	Note that $E'$ and $E''$ are compatible if and only if the value 
	of the maximum flow in the modified flow network is exactly~\pg.
\end{proof}

Our technique, based on flows, does not allow us to enforce 
a pairing of the edges in~$E'$ and in~$E''$.  With other words,
we cannot select edges $e'_1,e'_2\in E'$ and $e''_1,e''_2\in E''$, 
and insist that $e'_1$ and~$e''_1$ go to one page, and $e'_2$ and~$e''_2$
go to another page.  This difficulty
will play an important role in our forthcoming extension.

\subparagraph*{An XP algorithm for the general case.}
Let $H$ be a finite hitting set of~$(G,\sigma)$.
We assume, without loss of generality,
that $H \cap \sigma(V(G))=\emptyset$.
Given a subset~$X\subseteq H$, we say that an
edge~$(u,v)$ of~$G$ with $\sigma(u)<\sigma(v)$ \emph{bridges}~$X$ if
$\sigma(u) < \min X$, $\max X < \sigma(v)$, and $X$ is the largest
subset of $H$ with this property.  
For each $X\subseteq H$, let \emph{$E_X$} 
be the subset of edges of~$(G,\sigma)$ that bridge~$X$. For example, in
\cref{fig:3-hitting-set}, $|H|=3$, and the edges
in~$E_H$ lie in the outer gray region.

\begin{figure}[tb]
  \centering
  \includegraphics[page=1]{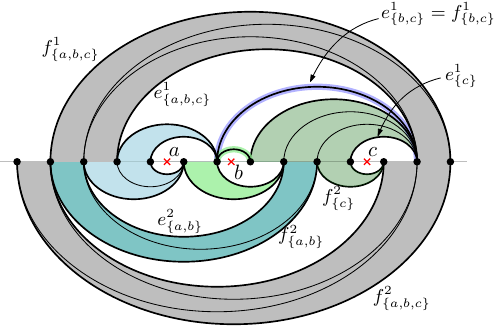}
  \caption{Encoding $\langle \mathcal{E}^1,\mathcal{E}^2 \rangle$ of a
    2-page drawing for an instance with hitting set~$H=\{a,b,c\}$ (red
    crosses).  For each $X \subseteq H$ and page $q \in [2]$, the
    edges $e_X^q$ and $f_X^q$ (if they exist) are thicker than the
    other edges.  Each colored region corresponds to a set of edges
    that bridge the same subset of~$H$.}
  \label{fig:3-hitting-set}
\end{figure}

Consider any drawing of a subgraph of $(G,\sigma)$ with edge set $\tilde E$ 
on $\pg$ pages without crossings. For each page $q \in [\pg]$, 
let \emph{$\tilde E^q$} be the set of edges in $\tilde E$ that are on page~$q$,
and let \emph{$\mathcal{X}^q$} be the family of subsets of~$H$ bridged by
some edge of $\tilde E^q$.  Since there are no
crossings on page~$q$, the sets of $\mathcal{X}^q$ form a so-called {\em laminar} 
family: any two sets in $\mathcal{X}^q$ are either disjoint or 
one contains the other. %
For each $X\in \mathcal{X}^q$, let \emph{$e_X^q$} be the smallest edge
of $\tilde E^q$ that bridges $X$, and let \emph{$f_X^q$} be the largest edge
of $\tilde E^q$ that bridges $X$; it may be that $e_X^q=f_X^q$. 
Note that for each $X,Y\in \mathcal{X}^q$ with 
$X\subsetneq Y$, the edge $e_Y^q$ contains $f_X^q$.
We say that the \emph{partial encoding} of~$\tilde E$ on page~$q$ is 
$\mathcal{E}^q= \{ (X,e_X^q, f_X^q) \mid X\in \mathcal{X}^q \}$ and the \emph{encoding} of~$\tilde E$ is $\langle \mathcal{E}^1,\dots,\mathcal{E}^p \rangle$.

When a set $X$ is bridged on only one page of an optimal drawing, 
say $X\in \mathcal{X}^1$, then we just have to select as many edges 
as possible without crossing from
those contained between $e^1_X$ and $f^1_X$, because the edges of $E_X$ cannot
appear in any other page. 
The challenge that we face is the following: when the same set $X$
appears in $\mathcal{X}^q$ for different $q\in [\pg]$, the choices
of which edges are drawn in each of those pages are not independent.
However, we can treat all such pages together, exchanging some parts
of the drawings from one page to another, as follows.
For each $X\subseteq H$, let \emph{$Q_X$}$\,=\{ q\in [\pg] \colon (X,e_X^q, f_X^q)\in \mathcal{X}^q\}$ be the set of pages where some edges bridge~$X$.

\begin{lemma}
  \label{lem:switching}
	Consider $\tilde E\subseteq E(G)$ that can be drawn in $\pg$ pages without
	crossings, and let $\langle \mathcal{E}^1,\dots, \mathcal{E}^\pg \rangle$ be the corresponding encoding.
	For every $X\subseteq H$ with $Q_X\neq \emptyset$, 
	let \emph{$\tilde E_X'$}$\,=\{ e_X^q\mid q\in Q_X,\, (X,e_X^q, f_X^q)\in \mathcal{E}^q\}$, let
	\emph{$\tilde E''_X$}$\,=\{ f_X^q\mid q\in Q_X,\, (X,e_X^q, f_X^q)\in \mathcal{E}^q \}$, and
	let~\emph{$F_X$} be the set of edges in~$E_X$ obtained when using 
	\cref{lem:1-hitting-set} for $\pg'=|Q_X|$ pages with boundary edges
	$\tilde E'_X$ and $\tilde E''_X$.
	Then the ordered subgraph with edge set $\bigcup_{X} F_X$
	can be drawn on \pg pages without crossings and 
	contains at least as many edges as $\tilde E$.
\end{lemma}
\begin{proof}
	Consider a fixed $X \subseteq H$ with $Q_X \ne \emptyset$.
	For each $q\in Q_X$, let \emph{$F_X^q$} be the set of edges in~$F_X$ 
	that appear on the same page as $e_X^q \in \tilde E'_X$ when using the algorithm
	of \cref{lem:1-hitting-set}. 
	Since each element of $\tilde E''_X$ is on a different page,
	let \emph{$\sigma$}$\colon Q_X \rightarrow Q_X$ be the permutation
	such that $f_X^{\sigma(q)}$ is the unique element of $\tilde E''_X$ in $F_X^q$.
	
	We make a drawing of $\hat E := (\tilde E\setminus E_X)\cup F_X$ on $\pg$
	pages by assigning edges to pages, as follows.
	For each $q\in [\pg]\setminus Q_X$, we just set $\hat E^q=\tilde E^q$. 
	For each $q\in Q_X$, let $\hat E^q$ be obtained from $\tilde E^{\sigma(q)}$
	by removing the edges contained in~$f_X^{\sigma(q)}$, adding the edges of~$F_X^q$,
	and adding the edges of~$\tilde E^q$ contained in~$e_X^q$.
    For an example, see \cref{fig:3-hitting-set-2}.
	For each $q$, the edges of~$\hat E^q$ can be drawn on a single page without crossings.
	This is obvious for $q\in [\pg]\setminus Q_X$.  For $q\in Q_X$,
	this is true because~$e^q_X$ and~$f_X^{\sigma(q)}$ act as shields between~$F_X^q$
	and the other two groups of edges, one containing $f_X^{\sigma(q)}$ and the
	other contained in $e^q_X$.

	Since $\tilde E\cap E_X = \big( \bigcup_{q\in Q_X} \tilde E^q \big)\cap E_X$ 
	is a feasible solution for the problem solved in \cref{lem:1-hitting-set},
	we have $|\tilde E\cap E_X|\le |F_X|$.  Therefore 
	\emph{$\hat E$}$\, = (\tilde E\setminus E_X)\cup F_X$ is at least as large as~$\tilde E$.

	\begin{figure}[tb]
	  \centering
	  \includegraphics[page=2]{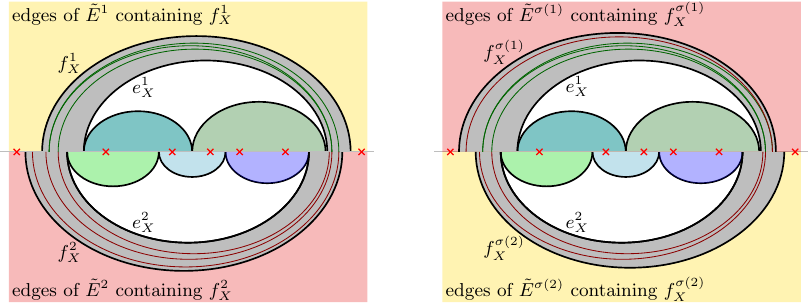}
	  \caption{Left: A 2-page drawing of~$\tilde E$.  The gray
            region corresponds to the set
            $\tilde E_X = \tilde E^1_X \cup \tilde E^2_X$ when $X$ is
            the set of the inner five red crosses.  Right: drawing of
            a set $\hat E = (\tilde E \setminus E_X)\cup F_X$ where
            $\sigma(1)=2$ and $\sigma(2)=1$.  Note that $\tilde E_X$
            and $\hat E_X$ can be different; namely if
            $F_X \ne \tilde E^1_X\cup \tilde E^2_X$.}
	  \label{fig:3-hitting-set-2}
	\end{figure}
	
	Summarizing: for a fixed $X$, we have converted $\tilde E$ into another set
	of edges $\hat E$ that is no smaller and can be drawn without crossings on \pg pages
	such that $F_X= \hat E\cap E_X$ and such that no edge outside $E_X$
	is changed (that is, $\tilde E\setminus E_X = \hat E\setminus E_X$).
	In general, the encoding $\langle \mathcal{E}^1,\dots, \mathcal{E}^\pg \rangle$ changes,
	but the sets $\tilde E_X', \tilde E''_X$ remain unchanged for every set~$X$.
	We now iterate this process for each $X \subseteq H$. 
    The last set $\hat E$ that we obtain is $\bigcup_{X} F_X$ because every edge of~$\tilde E$ is in $E_X$ for some $X \subseteq H$.
	The result follows.
\end{proof}

We now argue that, on a single page $q\in [\pg]$, 
the number of possible partial encodings~$\mathcal{E}^q$ is at most $m^{4h-2}$.
First note that $\mathcal{X}^q$ contains at most $2h-1$ sets: at most $h$ sets in $\mathcal{X}^q$ are inclusionwise minimal,
and any non-minimal element $X \in \mathcal{X}^q$ is obtained by 
joining two others.
This means that $\mathcal{E}^q$ is characterized by selecting
at most $4h-2$ edges $e_X^q$ and $f_X^q$, and such a selection
already determines implicitly the sets $\mathcal{X}^q$.
When considering all pages together, there are at most $m^{(4h-2)\cdot\pg}$
encodings $\langle \mathcal{E}^1,\dots,\mathcal{E}^\pg \rangle$,
and, for each $X \in \bigcup_{q\in [\pg]}\mathcal{X}^q$, we have to apply 
the algorithm of \cref{lem:1-hitting-set}, which takes
$\OO(|E_X|^3 \log n \log\log \pg)$ time.
Since the edge sets $E_X$ are pairwise disjoint for different $X\subseteq H$,
for each encoding we spend $\OO(m^3 \log n \log\log \pg)$ time.
Finally, we return the best among all encodings that give rise
to a valid drawing without crossings. 
Since the encoding of an optimal solution will be considered at least once,
\cref{lem:switching} implies that we find an optimal solution.
Therefore, the total running time is
$\OO(m^{(4h-2)\cdot \pg+3} \log n \log\log \pg)$.  
We summarize our result.

\begin{theorem}
  \label{thm:hitting-xp}
  \delppageplanar{\pg} is in XP 
  with respect to $h+\pg$.
\end{theorem}

\section{Multiple-Track Crossing Minimization}
\label{sec:multiple-track}

Let $G=(A \cup B,E)$ be a bipartite graph where all edges connect a
vertex of $A$ to a vertex of~$B$ and $A\cap B=\emptyset$.  
We further have a given linear order $\sigma_A$ for 
the vertices of $A$.  For the vertices of $B$ we
do not have any additional information or constraints.
In this section we consider \emph{spine+$\tr$-track drawings} of $G$, 
defined as follows:
\begin{itemize}
\item the vertices of $A$ are placed on a line $\ell_0$, 
  called \emph{spine}, in the order determined by $\sigma_A$;
\item the vertices of $B$ are placed on $\tr$ different 
  lines $\ell_1,\dots,\ell_\tr$ parallel to the spine; 
  each line~$\ell_q$ is placed on a different
  \emph{page} (half-plane) $\pi_q$ of a book;
\item all pages $\pi_1,\dots,\pi_\tr$ have $\ell_0$
  as a common boundary and are otherwise pairwise disjoint;
\item for each $q\in [\tr]$, the edges with endpoints in $\ell_0$
  and $\ell_q$ are drawn as straight-lines edges in the page $\pi_q$.
\end{itemize}
One can interpret this as a drawing in three dimension, 
as shown in \cref{fig:k-track}. Note that because the graph is bipartite
and each edge has a vertex in $A$ and a vertex in $B$, there are no edges
connecting two vertices in the spine, and in particular there are no ``nested'' edges.

To describe the drawing combinatorially, it suffices to partition
$B$ into sets $B_1,\dots,B_\tr$, one per line, and we have to decide for each $B_q$
the order $\sigma_{B_q}$ of the vertices $B_q$ along $\ell_q$.
The number of crossings of the drawing is the sum of the number of crossings within
each page, where the number of crossings within a page is the number of pairs of edges that cross each other.
The value $\CR_{\tr}((A,\sigma_A),B,E)$ is the minimum number 
of crossings over all spine+$\tr$-track drawings, and the purpose
of this section is to discuss its computation.

\begin{figure}[tb]
\centering
 \includegraphics[page=3]{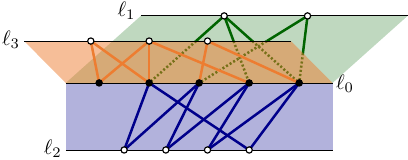}
 \caption{A spine+3-track drawing. In this example, $B_1$ has
	two vertices, $B_2$ has four vertices and $B_3$ has three vertices.
	The drawing has $2+5+2=9$ crossings.}
  \label{fig:k-track}
\end{figure}

We start discussing spine+$1$-track drawings and its 
corresponding value $\CR_1((A,\sigma_A),B,E)$.
See \cref{fig:1-track-1} for examples of drawings.
This is the minimum number of crossings in a two-layer drawing 
with the order on one layer, $A$ in this case, fixed.
We want to choose the order $\sigma_B$
that minimizes the number of crossings.  Let
$\CR_1((A,\sigma_A),(B,\sigma_B),E)$ be the crossing number for a fixed
order $\sigma_B$.
Then $\CR_1((A,\sigma_A),B,E)$ is the minimum of 
$\CR_1((A,\sigma_A),(B,\sigma_B),E)$ when we optimize over 
all orders $\sigma_B$ of $B$.
The obvious approach is to try all different possible orders
$\sigma_B$ of $B$, compute $\CR_1((A,\sigma_A),(B,\sigma_B),E)$ for
each of them, and take the minimum.  
This yields an algorithm with time
complexity $2^{\OO(n\log n)}$.  We improve over this trivial algorithm
as follows.

\begin{figure}[tb]
  \centering
  \includegraphics[page=1]{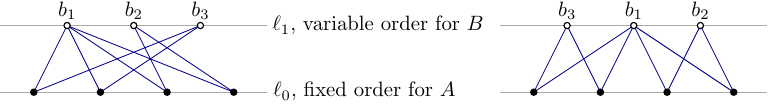}
  \caption{Two different orders $\sigma_B$ give different number of 
  crossings in the spine+$1$-track drawing: 
  10 on the left and 2 on the right.}
  \label{fig:1-track-1}
\end{figure}

\begin{theorem}
  \label{thm:1-track}
  We can compute $\CR_1((A,\sigma_A),B,E)$ in $\OO(2^n n)$ time, 
  where $n=|A|+|B|$.
\end{theorem}
\begin{proof}
  Construct a complete, directed, edge-weighted graph $H$ as follows:
  \begin{itemize}
  \item $V(H)=B$
  \item put all directed edges in $H$;
  \item the directed edge $(x,y)$ of $H$ gets weight
    $c_{x,y}=\CR_1((A,\sigma_A),(\{ x,y\},\sigma_{x,y}),E)$, where
    $\sigma_{x,y}$ is the order for $\{ x,y \}$ that places $x$ before
    $y$.
  \end{itemize}
  An ordering of $B$ corresponds to a Hamiltonian path in $H$.
  Consider any Hamiltonian path in $H$ defined by an order
  $\sigma_B$.  Since each crossing happens between two edges incident
  to different vertices of $B$, we have
  \begin{equation}\label{eq:cr2}
    \CR_1((A,\sigma_A),(B,\sigma_B),E) ~=~ \sum_{\begin{minipage}{2cm}\scriptsize\centering
        $x,y\in B$\\ $\sigma_B(x)<\sigma_B(y)$\end{minipage}} c_{x,y}
    ~=~ \sum_{x\in B}%
    \sum_{\begin{minipage}{2cm}\scriptsize\centering
          $y \in B$\\ $\sigma_B(x)<\sigma_B(y)$\end{minipage}} c_{x,y}.%
  \end{equation}
  With this interpretation, the task is to find in $H$ a Hamiltonian
  path such that the sum of the $c_{\cdot,\cdot}$-weights from each
  vertex to all its successors is minimized.  This problem is amenable
  to dynamic programming across subsets of vertices, as it is done for
  the \textsc{Traveling Salesperson Problem}; see \cite{Bellman62} or
  \cite[Section 6.6]{DasguptaPV-book}.

  We define a table $C$ by setting, for each $X\subseteq B$,
  \[
    C[X] ~=~ \CR_1((A,\sigma_A),X,\{ (a,x)\in E \mid x\in X, a \in A \}).
  \]
  Then $C[X]$ is the number of crossings when we remove the vertices
  $B\setminus X$ from $H$.  We are interested in computing 
  $C[B]$ because $C[B]= \CR_1((A,\sigma_A),B,E)$.
	
  We obviously have $C[X]=0$ for each $X\subseteq B$ with $|X|\le 1$.
  Whenever $|X|> 1$, we use~\eqref{eq:cr2} and the definition of
  $C[X]$ to obtain the recurrence
  \begin{equation}\label{eq:recurrencecr}
    C[X] ~=~ \min_{y\in X} \left(C[X\setminus \{y\}] + \sum_{x\in X\setminus
      \{y\} } c_{x,y}\right).
  \end{equation}
  The proof of this is a standard proof in dynamic programming, where
  $y$ represents the last vertex of $X$ in the ordering\arxiv{;
  see \cref{fig:2tracks-2}.

  \begin{figure}[tb]
    \centering
	  \includegraphics[page=2]{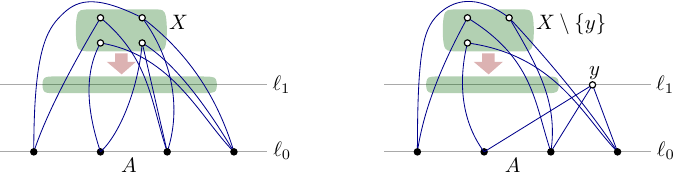}
	  \caption{Schema showing $C[X]$ and what happens when the last 
	   vertex of $X$ gets fixed.}
  	 \label{fig:2tracks-2}
  \end{figure}
  }{.}
  
  Each value $c_{x,y}$ can be computed in
  $\OO(\deg_G(x)+\deg_G(y))=\OO(n)$ time, which means that, over all
  pairs $(x,y)$, we spend $\OO(n^3)$ time.  Each value
  $\eta[X,y]:=\sum_{x\in X\setminus \{y\} } c_{x,y}$, defined for
  $X\subseteq B$ and $y\in X$, can be computed for increasing values of
  $|X|$ in constant time per value by noting that
  \[
    \text{for every $X\subseteq B$ and distinct $y,z\in X$:}~~~ \sum_{x\in
      X\setminus \{y\} } c_{x,y} = c_{z,y} + \sum_{x\in X\setminus
      \{y,z\} } c_{x,y}.
  \]
  Therefore, we compute the value $\eta[X,y]$ for every $X\subseteq B$
  and $y\in X$ in
  $\Theta\left(\sum_{k=0}^n\binom{n}{k}k \right)= \Theta(2^n n)$ total time.
  (The direct computation using the sums anew for each value would
  take
  $\Theta\left( \sum_{k=0}^n\binom{n}{k}k^2 \right)= \Theta(2^n n^2)$,
  which is strictly larger.)

  After this we can compute $C[X]$ for increasing values of $|X|$
  using the recurrence of \cref{eq:recurrencecr},
  which means that we spend $\OO(|X|)$ time for each $X$. 
  This step also takes $\OO(2^n n)$ time for all $X$.  
  Finally we return $C[B]$. An optimal solution
  can be recovered using standard book-keeping techniques.
\end{proof}

Now we consider the case of arbitrary track number~$\tr$.

\begin{theorem}
  \label{thm:k_tracks}
  We can compute $\CR_{\tr}((A,\sigma_A),B,E)$ in $2^n n^{\OO(1)}$ time 
  for every $\tr>1$, where $n=|A|+|B|$.
  For $\tr=1$ and $\tr=2$, the value can be computed in $O(2^n n)$ time.
\end{theorem}
\begin{proof}
  Once we fix a set $B_q$ for the $q$th page, we can optimize the order
  $\sigma_{B_q}$ independently of all other decisions.
  Therefore, we want to compute 
  \[
	\min~~\sum_{q=1}^\tr \CR_1((A,\sigma_A),B_q,E_q),
  \]
  where $E_q$ is the set of edges connecting vertices from $A$ to $B_q$,
  and where the minimum is only over all the partitions 
  $B_1,\dots,B_\tr$ of $B$.
  
  As we did in the proof of Theorem~\ref{thm:1-track},
  for each subset $X\subseteq B$, we define 
  \[
    C[X] ~=~ \CR_1((A,\sigma_A),X,\{ (a,x)\in E \mid x\in X, a \in A \}).
  \]
  In the proof of Theorem~\ref{thm:1-track} we argued
  that the values $C[X]$ can be computed in $\OO(2^n n)$ time
  for all $X\subseteq B$ simultaneously.

  We have to compute now 
  \[
	\min \left\{ \sum_{q=1}^\tr C[B_q] \colon
		\text{$B_1,\dots,B_\tr$ is a partition of $B$} \right\}.
  \]
  The case of $\tr=1$ has been covered in Theorem~\ref{thm:1-track}.
  For $\tr=2$, we have to compute 
  \[
	\min \left\{ C[B_1]+ C[B\setminus B_1] \mid 
		\text{$B_1\subseteq B$} \right\},
  \]
  which can be done in $O(2^n)$ additional time iterating
  over all subsets $B_1$ of $B$.

  For $\tr>2$, we use the algorithm of 
  Bj{\"{o}}rklund et al.~\cite{BjorklundHKK07} for subset convolution,
  as follows.
  Define for each $X\subseteq B$ and for $q\in[\tr]$ the 
  ``entry table''
  \begin{align*}
    T[X,q] ~&=~ \CR_q((A,\sigma_A),X,\{ (a,x)\in E \mid x\in X, a \in A \}) \\
		~&=~ \min \left\{ C[B_1]+\ldots+C[B_q] \mid 
		\text{$B_1,\dots,B_q$ is a partition of $X$} \right\}.
  \end{align*}
  We obviously have $T[X,1] = C[X]$ for all $X$.
  For $q>1$, we have the recursive relation
  \begin{align*}
	T[X,q] ~&=~ \min \left\{ T[Y,q-1] + C[X\setminus Y] \mid 
		Y\subseteq X \right\}.
  \end{align*}
  Therefore, for $q>1$, the function $X\mapsto T[X,q]$ is, 
  by definition, the subset convolution of the functions 
  $X\mapsto T[X,q-1]$ and $X\mapsto C[X]$ in the $(\min,+)$ ring. 
  These functions take integer values
  on $\{0,\dots, n^4\}$ because $n^4$ is an upper bound for 
  $\CR_q((A,\sigma_A),B,E)$ for any~$q$. 
  It follows from~\cite{BjorklundHKK07} that one can obtain 
  $T[X,q]$ for all $X\subseteq B$ in $2^n n^{\OO(1)}$ time,
  assuming that $T[\cdot,q-1]$ and $C[\cdot]$ are already
  available. We compute the entries $T[\cdot,q]$ 
  for $q=2,\dots, \tr$, which adds a multiplicative
  $\tr\le n$ to the final running time.  
\end{proof}

Using the theorem for increasing values of~\tr, we obtain the following.

\begin{corollary}
  \label{cor:min-track}
  We can compute the smallest value $\tr$ such that 
  $\CR_{\tr}((A,\sigma_A),B,E)=0$ in $2^n \cdot n^{\OO(1)}$ time,
  where $n=|A|+|B|$.
\end{corollary}

\section{Open Problems}
\label{sec:open}

\begin{enumerate}
\item Could we use the concept of the conflict graph for other
  crossing reduction problems?
\item Is \delpplanar\ $W[1]$-hard with respect to the natural
  parameter~$k$ if \pl is part of the input?  Can we reduce from
  \textsc{Independent Set}?  Note that \deldegp is $W[1]$-hard with
  respect to treewidth \cite{Betzler2012} and that outer-\pl planar
  graphs have treewidth $\OO(\pl)$ \cite{ckllw-bo-GD17} (which also
  follows from \cref{lem:separator}).
\item What if the vertex order is not given? In other words, what is
  the parameterized complexity of edge deletion to outer-\pl planarity?
\item What about exact algorithms for computing the crossing number of
  an ordered graph?  As Masuda et al.~\cite{MasudaNKF90} showed, the
  problem is NP-hard for two pages.  %
  In their NP-hardness reduction, they use a large number of crossings, and it is easy to get an algorithm that is exponential in the number of edges;
  see \cref{thm:fopn_exact}.
  Can we get a running time of $2^{n}\cdot n^{\OO(1)}$ or perhaps even subexponential in $n$?
  Recall that the algorithm of Liu et al.~\cite{LiuCHW21} checks in $n \cdot (\CR + 2)^{\OO(\pw^2)}$ time whether a graph with %
  pathwidth~$\pw$ can be drawn on a given number of pages with at most $\CR$ crossings in total.
\end{enumerate}

\bibliographystyle{abbrvurl}
\bibliography{gdexact}

\end{document}